\documentclass[a4paper,11pt]{article}
\usepackage{amsmath}
\usepackage{amssymb}
\usepackage{authblk}

\usepackage{hyperref}
\usepackage{xcolor}
\hypersetup{
    colorlinks,
    linkcolor={blue!50!black},
    citecolor={green!50!black},
    urlcolor={blue!80!black}
}

\newcommand{\defeq}{\stackrel{\rm def}{=}}

\newcommand{\prob}{\mbox{\bf{Pr}}}

\newcommand{\distr}{\mbox{$\cal{D}$}}

\newtheorem{theorem}{Theorem}

\newtheorem{lemma}[theorem]{Lemma}
\newtheorem{proposition}[theorem]{Proposition}

\newtheorem{claim}[theorem]{Claim}
\newtheorem{definition}[theorem]{Definition}

\newcommand{\bproof}{\noindent{\it Proof}}
\newcommand{\eproof}{\hspace*{\fill}$\Box$~~~~~\bigskip}
\newenvironment{proof}{\bproof. }{\eproof}

\newenvironment{psketch}{\bproof~{\it Sketch}. }{\eproof}

\title{A note on faithful coupling of Markov chains}

\author[1]{\small Debojyoti Dey}
\author[2]{\small Pranjal Dutta}
\author[1]{\small Somenath  Biswas} 
\affil[1]{\footnotesize Department of Computer Science and Engineering, Indian Institute of Technology Kanpur, Kanpur-208016.}
\affil[2]{\footnotesize Chennai Mathematical Institute, Chennai-603103.}

\date{October 2017}
\begin{document}
\maketitle
\begin{abstract}
One often needs to turn a coupling $(X_i, Y_i)_{i\geq 0}$ of a Markov chain into a {\em sticky coupling} where once
$X_T = Y_T$ at some $T$, then from then on, at each subsequent time step $T'\geq T$, we shall have $X_{T'} = Y_{T'}$. 
However, not all of what are considered couplings in literature, even Markovian couplings, can be turned into
sticky couplings, as proved by Rosenthal \cite{rosenthal} through a counter example. Rosenthal 
then proposed a strengthening of the Markovian coupling notion, termed
as {\em faithful coupling}, from which a sticky coupling can indeed be obtained. We identify the reason why
a sticky coupling could not be obtained in the counter example of Rosenthal, which motivates us to define a type of
coupling which can obviously be turned into a sticky coupling. We show then that the new type of coupling that
we define, and the  faithful coupling as defined by Rosenthal \cite{rosenthal}, are actually identical. 
Our note may be seen as a demonstration of the naturalness of the notion of faithful coupling.
\end{abstract}       

\section{Introduction}
First, we recall certain basic definitions 
and fix our notations.
For a random variable $X$, we denote its distribution as $\distr(X)$.
Let $M$ be a finite Markov chain on the state space $\Omega$, and with the transition matrix $P$. Following
\cite{norris}, we abbreviate the Markov chain 
$(X_i)_{i\geq 0}$, with $\distr(X_0) = \lambda$ and $\distr(X_{i+1}) = \distr(X_i)P$ for $i\geq 0$, as Markov$(\lambda, P)$.
A {\em coupling} of the Markov chain
$M$ is a process $(X_i, Y_i)_{i\geq 0}$, with each $X_i, Y_i$ taking values in $\Omega$ and evolving on a common probability space, where $(X_i)_{i\geq 0}$ 
satisfies the property that, for $i\geq 0$, $\distr(X_{i+1}) = \distr(X_i)P$ and similarly, 
$(Y_i)_{i\geq 0}$ satisfies, for $i\geq 0$,  $\distr(Y_{i+1}) = \distr(Y_i)P$.

{\small\bf Remark A:} {\small The definition of Markov chain coupling, as given in \cite{peres}, says that both $(X_i)_{i\geq 0}$
and $(Y_i)_{i\geq 0}$ are Markov chains with transition matrix $P$.  The reason we prefer the above weaker definition is
because in literature we find processes considered as couplings that do not satisfy the stronger notion of \cite{peres}, but
do satisfy our above definition. A Markov chain $(X_i)_{i\geq 0}$ with transition matrix $P$, by definition, satisfies the property ($\pi$) that  
for $i\geq 0$, conditioned on $X_i = l, X_{i+1}$ has the distribution $P(l, \cdot)$ and is indepedent of 
$X_0, \ldots, X_{i-1}$. As we shall see in the next section, there is a Markovian coupling $(X_i, Y_i)_{i\geq 0}$, where
$(X_i)_{i\geq 0}$ does not satisfy the property ($\pi$), although $X_i$'s do evolve (under certain conditions) as per
the transition matrix $P$.}
  
Let $\distr(X_0)$ be $\mu$ and $\distr(Y_0)$ be $\nu$. 
In the above, since each pair, $X_i$ and $Y_i$, $i\geq 0$, is
defined on the same probability space, the two random variables, $X_i, Y_i$, for each $i\geq 0$, are coupled,
and so we have, from the above definition of Markov chain coupling, and Proposition 4.7 of \cite{peres}:
\begin{equation}
\label{coupling-1}  
\|\mu P^i - \nu P^i\|_{\rm TV} = \|\distr(X_i) - \distr(Y_i)\|_{\rm TV} \leq \prob(X_i \neq Y_i)
\end{equation}
(Here, $\|\cdot\|_{\rm TV}$ denotes the total variational distance between two distributions.)
 
Let us suppose that a coupling of $M$ has the {\em `now-equals-forever'} property (\cite{rosenthal}), namely,
if $X_n = Y_n$ for some $n$, then for all $j\geq n, X_j = Y_j$.
Then, if $T$ is defined as the random time
\begin{equation}
\label{T-defn} 
T\defeq \inf\{i| X_i = Y_i\}
\end{equation}
we get from the relation (\ref{coupling-1})
\begin{equation}
\label{distance-bound} 
|\|\distr(X_i) - \distr(Y_i)\|_{\rm TV} \leq \prob(X_i \neq Y_i) = \prob(T > i)
\end{equation}
as the two events $(X_i \neq Y_i)$ and $(T > i)$ imply each other, making use of the `now-equals-forever'propery.
     
As such, the applicability of the inequality (\ref{distance-bound}) is limited as most couplings of Markov chains will
not have the `now-equals-forever' property. The way that is often resorted to to tackle this problem 
is to define a new process $(Z_i)_{i\geq 0}$ 
as follows: (\cite{rosenthal})
\begin{equation}
\label{sticking}
Z_i = \left\{ \begin{array}{ll}
       X_i  & \mbox{ if $i \leq T$}\\
       Y_i  & \mbox{ otherwise}
              \end{array} 
       \right.
\end{equation}
where $T$ is as in (\ref{T-defn}) above.
If it so happens that $(Z_i)_{i\geq 0} = \mbox{Markov}(\mu, P)$ then, since $(Z_i, Y_i)_{i\geq 0}$ is a coupling of
the Markov chain satisfying `now-equals-forever' property, 
%
%
we have
\[\|\mu P^i - \nu P^i\|_{\rm TV} = \|\distr(Z_i) - \distr(Y_i)\|_{\rm TV} \leq \prob(T > i)\]
The construction of $(Z_i)$ has been evocatively termed as {\em sticking} and the resultant new coupling
as a {\em sticky} coupling \cite{hirscher}.
\begin{definition}[Markovian coupling of a Markov chain]
A coupling $(X_i, Y_i)_{i\geq 0}$ of a Markov chain $M$ on state space $\Omega$, and with its transition 
matrix as $P$, where $\distr(X_0) = \mu$ and
$\distr(Y_0) = \nu$, for given $\mu$ and $\nu$, is said to be a {\em Markovian coupling}
if $(W_i)_{i\geq 0}$, with $W_i\defeq (X_i, Y_i)$ is itself a Markov chain on state space $\Omega\times\Omega$,
with a specified joint distribution of $\distr(X_0)$ and $\distr(Y_0)$ as the distribution $\distr(W_0)$. In
other words, there is a $(\Omega\times\Omega, \Omega\times\Omega)$ transition matrix, say $Q$, such that
for {\em some} distribution $\theta_0$ on $\Omega\times\Omega$, satisfying that its marginals are $\mu = \distr(X_0)$
and $\nu = \distr(Y_0)$, we will have, for $i\geq 1$, with $\theta_i\defeq\theta_{i-1}Q$, the two marginals of
$\theta_i$ will be $\mu P^i$, that is, $\distr(X_i)$ 
and $\nu P^i$, that is $\distr(Y_i)$.
\end{definition}

{\small\bf Remark B:} {\small We note that the defintion above is specific to a joint distribution of $(X_0, Y_0)$.
The reason
for the specificity is that in 
the example given by Rosenthal (\cite{rosenthal}) of a
Markovian coupling for which sticking fails, as we note in the next section, the ability to evolve the two copies 
as per the transition matrix $P$ crucially depends on the
initial coupling of $X_0$ and $Y_0$.}

As noted in \cite{hirscher}, it has been stated at times that sticking will work for all Markovian couplings. However,
this is not correct, Rosenthal \cite{rosenthal} provides a counter example. He then provides a stronger version of Markovian
couplings, termed as {\em faithful couplings}, for which it is proved that sticking will provably result in a coupling 
of $M$ where $(Z_i)_{i\geq 0}$, as defined in (\ref{sticking}) above, will indeed be Markov$(\mu, P)$.   
We discuss Rosenthal's counterexample in Section \ref{sec2} and identify what we consider to be the reason why 
the sticking operation
fails there. This reason then motivates us in defining a stronger version of Markovian coupling (Section \ref{sec3}) for 
which
it is easy to see that the sticking will indeed work. We then prove that the new notion of
coupling that we define is actually equivalent to the notion of faithful coupling.

\section{Rosenthal's counterexample}
\label{sec2}
In his example, Rosenthal \cite{rosenthal} considers the Markov chain $M$ with state space $\Omega = \{0,1\}$, 
and with the transition matrix $P$ defined as:      
\[P = \begin{array}{c|cc}
            &0 &1 \\
            \hline
         0  &1/2  &1/2\\
         1  &1/2  &1/2
       \end{array}\] 
$(X_i), (Y_i),{i\geq 0}$  are defined as follows:
\begin{enumerate}
\item Each of $\distr(X_0), \distr(Y_0)$ is the uniform distribution on the state space $\{0,1\}$,
\item Each $\distr(Y_i), i > 0$ is also the uniform distribution, whereas for $i \geq 0, X_{i+1} \defeq X_i\oplus Y_i$,
$\oplus$ being the exclusive-or operation.    
\end{enumerate}
Thus, the joint evolution of the two copies $(X_i, Y_i)$ of the chain, on the state space $\Omega\times\Omega$, that is,
$\{(0,0), (0,1), (1,0), (1,1)\}$ is governed by the following transition matrix $Q$:
\[Q = \begin{array}{c|cccc}
              &(0,0) &(0,1) &(1,0) &(1,1)\\
        \hline
       (0,0)  &1/2   &1/2    &0    &0\\
       (0,1)  &0     &0      &1/2  &1/2\\  
       (1,0)  &0     &0      &1/2  &1/2\\  
       (1,1)  &1/2   &1/2    &0    &0
      \end{array}\] 
{\bf Convention:} The distribtion of a random variable $X$ on $\Omega = \{0,1\}$ is specified as the vector
\[[\prob(X = 0)\ \ \  \prob(X = 1)]\] 
and similarly, the distribution of a random variable $W$ on $\Omega\times\Omega$ as the vector\\
\[[\prob(W = (0,0))\ \ \  \prob(W = (0,1))\ \ \  \prob(W = (1,0))\ \ \  \prob(W = (1,1))]\]

We note that for any distribution $\sigma\defeq [p\ \ \ 1-p]$,  $\Omega = \{0,1\}$, $\sigma P$ will be the uniform
distribution. We also note that for the trivial coupling $\theta$, of $[1/2\ \ \ 1/2]$ and $[1/2\ \ \ 1/2]$, 
which is $[1/4\ \ \ 1/4\ \ \ 1/4\ \ \ 1/4]$, $\theta Q$ will again be $\theta$. Therefore, with the joint
distribution  $\theta$ of
$\distr(X_0)$ and $\distr(Y_0)$, $Q$ indeed defines a Markovian coupling of the Markov chain $M$.
 
{\small\bf Remark C:} {\small We also note that there are other joint distributions
of the same $\distr(X_0)$ and $\distr(Y_0)$, for which $Q$ will {\em not} be a Markovian coupling: consider
$\theta'\defeq [3/8\ \ \ 1/8\ \ \ 1/8\ \ \ 3/8)$, $\theta'Q = [3/8\ \ \ 3/8\ \ \ 1/8\ \ \ 1/8]$. The marginals of 
the latter distributions are 
$[3/4\ \ \ 1/4]$ and $[1/2\ \ \ 1/2]$
respectively. Thus, if we use the joint distribution $\theta'$ of the same two initial distributions, while $Q$ evolves $Y_i$'s 
(understandably) as per $P$, it does not evolve $X_i$'s as per $P$. Hence with the initial joint distribution $\theta'$,
$Q$ fails to be a Markovian coupling of $M$. This is why we felt that it is necessary to make explicit the initial
joint distribution in the definition of Markovian coupling of a Markov chain.}


Rosenthal proves that although $Q$ is a Markovian coupling of $M$ with $\theta$ as the initial joint distribution, the
result of the sticking operation will fail to evolve $M$ correctly: with $T$ defined as in (\ref{T-defn}) and $Z_i$'s as 
defined in (\ref{sticking}), 
\begin{align*}
&\prob(Z_0 = 1, Z_1 = 0) \\
&= \prob(T = 0, Y_0 = 1, Y_1 = 0, X_0 = 1) + \prob(T > 0, X_0 = 1, Y_0 = 0, X_1 = 0)\\
&= 1/8 + 0\\
&= 1/8
\end{align*}
However, had $Z_i$'s been evolving as per $P$ with the same initial distribution as the uniform distribution, the
probability above would have been $1/4$, and not $1/8$. Hence, sticking fails. 

We see that sticking failed with $T = 0$. As it happens, sticking will fail here for every value of $T$ (except, of
course, for $T = \infty$). The reason is provided by the Proposition below. We need the following definition:
\begin{definition}
\label{delta-distribution} 
For a state space $\Omega$ with
$s$ an element in it, $\delta_s$ denotes the distribution on $\Omega$ in which the probability of $s$ is $1$, and
(therefore,) all the other states have probability $0$.  
\end{definition}
\begin{proposition}
\label{condition}
Let $M$ be a Markov chain with $\Omega$ as its state space and $P$ as its transition matrix, and let $Q$ be a transition
matrix for the state space $\Omega\times\Omega$ that defines a Markovian coupling $(W_i = (X_i, Y_i))_{i\geq 0}$ of $M$
where $\distr(X_0) = \mu$ and $\distr(Y_0) = \nu$, where the initial joint distribution used is $\theta_0$. This Markovian
coupling of $M$ can be turned into a sticky coupling using the sticking operation if the following condition holds: for every 
$s\in \Omega$, defining $\eta_0$ as the unique
joint distribution on $\Omega\times\Omega$ with $\delta_s$ and $\delta_s$ as its two marginals, and further defining $\eta_{i+1}$
as $\eta_i Q$, for each $i\geq 0$, we have that for each $\eta_i, i\geq 0$, will be a joint distribution that has $\delta_s P^i$
and $\delta_s P^i$ as its marginals. On the other hand, if $Q$ does not satisfy the condition, then the $Q$-defined Markovian
coupling, in general, cannot be turned into a sticky coupling.
\end{proposition}

\begin{psketch}
The condition ensures that both 
$(X_i)_{i\geq 0}$ and $(Y_i)_{i\geq 0}$ of the Markovian coupling of $M$ satisfy the 
{\em strong Markov property} (Theorem 1.4.2, \cite{norris}) 
with respect to the stopping time $T$, as defined in (\ref{T-defn}). Therefore, 
conditioned on $T = m$ and $X_T = Y_T = s$, 
each of $(X_{T+n})_{n\geq 0}$ and $(Y_{T+n})_{n\geq 0}$ will be Markov$(\delta_s, P)$
and the former will be 
independent of $X_0, X_1,\ldots, X_T$ and the latter will be independent of $Y_0, Y_1,\ldots, Y_T$. 
Thus, $(X_{T+n})_{n\geq 0}$ and $(Y_{T+n})_{n\geq 0}$ can replace each other.
Because of
this, sticking is guaranteed to work as $(X_i)_{i\geq 0}$ and $(X_0, X_1,\ldots, X_T, Y_{T+1},\ldots)$ will be
indistinguishable. 

We now show that the counter example of Rosenthal does not satify the condition as stated in the statement of the proposition. 
As we have seen that the Markovian coupling there cannot be turned into a sticky coupling, the counter example proves the
second part of the Proposition.
Let us suppose that $T = m$ and $X_T = Y_T = 0$. Conditioned on these events, $\distr(X_T) = \distr(Y_T) = [1\ \ \ 0]$.
The only joint distribution with these two distributions as marginals is $[1\ \ \ 0\ \ \ 0\ \ \ 0]$. The joint distribution of
$X_{T+1}$ and $Y_{T+1}$ will be given by $[1\ \ \ 0\ \ \ 0\ \ \ 0]Q$, which is $[1/2\ \ \ 1/2\ \ \ 0\ \ \ 0]$. 
This gives $\distr(X_{T+1})$
as $[1\ \ \ 0]$, which is not $[1/2\ \ \ 1/2] =  \distr(X_T)P$. Thus, $Q$ fails to ensure that, conditioned as above,
$\distr(X_{T+1}) = \delta_0P$. 
On the other hand, $\distr(Y_{T+1}) = \delta_0P$, and so   
$(X_{T+n})_{n\geq 0}$ and $(Y_{T+n})_{n\geq 0}$ are different and {\em cannot} replace each other, as required for the
sticking operation to work.
\end{psketch}
 
\section{Strong Markovian coupling}
\label{sec3}
We saw in Section \ref{sec2} that sticking fails for $Q$ because it does not evolve $(X_i)_{i\geq 0}$ as per $P$ for
certain joint distributions. Therefore, the following definition of a type of coupling of Markov chains immediately 
suggests itself to correct the defect:
\begin{definition}[Strong Markovian coupling]
A coupling of a finite Markov chain $M$ with state space $\Omega$ and transition matrix $P$ is a strong Markovian
coupling if there is a transition matrix $Q$ from $\Omega\times\Omega$ to $\Omega\times\Omega$ such that for every
pair of distributions $\mu$ and $\nu$, each on $\Omega$, and for every joint distribution $\theta$ of $\mu$ and $\nu$,
$\theta Q$ will be a joint distribution of $\mu P$ and $\nu P$.
\end{definition}
From this definition it follows that:
\begin{claim}
\label{strong-markov-properties}
Let $Q$ be a strong Markovian coupling of a finite Markov chain $M$ with state space $\Omega$ and
transition matrix $P$, and for two random variables $X_0$ and $Y_0$ on $\Omega$ let $\distr(X_0), \distr(Y_0)$   
be $\mu$ and $\nu$ respectively, and let $\theta_0$ be any joint distribution of $\mu$ and $\nu$. If we define
$\theta_{i+1}$ as $\theta_iQ$, for $i\geq 0$, and $X_{i+1}$ and $Y_{i+1}$ as two random variables whose distributions
are respectively the two marginals of $\theta_{i+1}$ then 
\begin{enumerate}
\item $(X_i)_{i\geq 0}$ will be Markov$(\mu, P)$ and $(Y_i)_{i\geq 0}$ will be Markov$(\nu, P)$. Consequently,
\item $\distr(X_i) = \mu P^i$ and $\distr(Y_i) = \nu P^i$, for $i\geq 0$,
\item Both $(X_i)_{i\geq 0}$ and $(Y_i)_{i\geq 0}$ will satisfy the strong Markov property, and therefore,
the coupling $Q$ of $M$ can be turned into a valid sticky coupling through sticking.
\end{enumerate}
\end{claim} 
We see next that the strong Markovian coupling is actually equivalent to the familiar notion of {\em faithful coupling}
which is defined as:
\begin{definition}[Faithful coupling \cite{rosenthal}]
A Markov coupling of a finite Markov chain $M$, with $\Omega$ as its state space and $P$ as its transition matrix, is a 
faithful coupling given by a Markov chain $(W_i)_{i\geq 0} = (X_i, Y_i)_{i\geq 0}$, on state space $\Omega\times\Omega$,
with transition matrix $Q$, if $Q$ satisfies, for all $i, j, i', j'\in\Omega$, the following
\[\sum_{j'\in\Omega}Q((i,j), (i',j')) = P(i,i'),\mbox{ and }\]
\[\sum_{i'\in\Omega}Q((i,j), (i',j')) = P(j,j')\] 
Equivalently, for all $t\geq 0$ and all $i, j, i', j'\in\Omega$
\[\Pr(X_{t+1} = i'|X_t = i, Y_t = j) = P(i,i'), \mbox{ and }\]
\[\Pr(Y_{t+1} = j'|X_t = i, Y_t = j) = P(j,j')\]
\end{definition}

{\small\bf Remark D:} {\small This kind of coupling has been termed in \cite{peres}, as well as in \cite{mitzen}, as
Markovian coupling; however, we follow the terminology of \cite{hirscher} for reasons given therein.
It is to be noted that most coupling constructions of Markov chains turn out to be faithful couplings, as the various
coupling examples in \cite{peres} demonstrate.}
\begin{proposition}
A Markov chain coupling is faithful if and only if it is a strong Markovian coupling.
\end{proposition}
The two lemmas below prove the two directions of the above proposition.
\newpage
\begin{lemma} 
Every faithful coupling of a Markov chain $M$ is a strong Markovian coupling.  
\end{lemma}
\begin{proof}
Let $M$ be a Markov chain with state space $\Omega$ and transition matrix $P$. Let the transition matrix $Q$, giving
transition probabilities from $\Omega\times\Omega$ to $\Omega\times\Omega$, define a faithful coupling of the Markov
chain $M$. Let $\mu$ and $\nu$ be two distributions on $\Omega$, and let $\theta$ be a joint distribution of
$\mu$ and $\nu$. We need to prove that $\mu P$ and $\nu P$ are the two marginals of $\theta Q$. In particular, we
need to show:
\begin{equation}
\label{for-x}
\mbox{For all } x\in\Omega, \sum_{y\in\Omega}(\theta Q)(x,y) = (\mu P)(x)
\end{equation}
and,
\begin{equation}
\label{for-y}
\mbox{For all } y\in\Omega, \sum_{x\in\Omega}(\theta Q)(x,y) = (\nu P)(y)
\end{equation}
We prove (\ref{for-x}) as follows:
\begin{align*}
&\sum_{y\in\Omega}(\theta Q)(x,y)\\
&=\sum_{y\in\Omega}\sum_{(u,v)\in\Omega\times\Omega}\theta(u,v)Q((u,v), (x,y))\\
&=\sum_{(u,v)\in\Omega\times\Omega}\theta(u,v)\sum_{y\in\Omega}Q((u,v), (x,y))\\
&=\sum_{(u,v)\in\Omega\times\Omega}\theta(u,v)P(u,x), \mbox{ because $Q$ defines a faithful coupling}\\
&=\sum_{u\in\Omega}P(u,x)\sum_{v\in\Omega}\theta(u,v)\\
&=\sum_{u\in\Omega}P(u,x)\mu(u), \mbox{ $\theta$ being the joint ditribution of $\mu$ and $\nu$}\\
&=(\mu P)(x)
\end{align*}
In a similar manner we can prove (\ref{for-y}).
\end{proof}

For the other direction, we prove
\begin{lemma}
Every strong Markovian coupling of a Markov chain $M$ is a faithful coupling of $M$.
\end{lemma}
\begin{proof} 
Let $M$ be a Markov chain with state space $\Omega$ and transition matrix $P$. Let the transition matrix $Q$, giving
transition probabilities from $\Omega\times\Omega$ to $\Omega\times\Omega$, define a strong Markovian coupling of the Markov
of $M$. For any $(u, v)\in\Omega\times\Omega$, consider the probabilty distribution $\delta_{(u,v)}$. This distribution is the
joint distribution $\delta_u$ and $\delta_v$, both defined on $\Omega$. (The definition of $\delta_x$ as in Definition 2.)
$\delta_uP$ is the $u$th row of $P$, namely, $P(u,\cdot)$. Similarly,   
$\delta_vP$ will be the $v$th row of $P$, namely, $P(v,\cdot)$, and $\delta_{(u,v)}Q$ will be the $(u,v)$th row of $Q$,
that is, $Q((u,v),\cdot)$. As $Q$ defines a strong Markovian coupling of $M$, we have that $Q((u,v),\cdot)$ will be the
joint distribution of $P(u,\cdot)$ and $P(v,\cdot)$. From this condition, we get for any $x, u, v\in\Omega$:
\begin{equation}
\label{faithful-1}
\sum_{y\in\Omega}Q((u,v),(x,y)) = P(u,x)
\end{equation}
Similarly, for any $y, u, v\in\Omega$, we get   
\begin{equation}
\label{faithful-2}
\sum_{x\in\Omega}Q((u,v),(x,y)) = P(v,y)
\end{equation}
As (\ref{faithful-1}) and (\ref{faithful-2}) are precisely the conditions to be met by $Q$ to be a faithful coupling, 
we conclude
that $Q$ defines a faithful coupling of $M$.
\end{proof}
\section{Concluding remarks}
We have seen that it is sufficient for a Markovian coupling to have the faithfulness property in order
to be turned to a sticky 
coupling. Is the faithfulness property also a {\em necessary} property? It may be that there is a Markovian coupling
of a Markov chain $M$ which evolves two copies of the chain correctly for a pair of initial distributions $\mu, \nu$ of
$M$, using a 
joint distribution of the pair, and satisfies the condition in the statement of Proposition \ref{condition}, but it either
does not work for some other joint distribution of the same pair $\mu$ and $\nu$, or, does not satisfy the condition for
some other pair of initial distributions of $M$. Such a Markovian coupling can be turned into a sticky coupling for the
$\mu, \nu$ pair, but will not be strongly Markovian, and hence will not be a faithful coupling of $M$. One feels that 
even if such an example exists, it is unlikely to be a natural example. Faithfulness appears to be the only natural
strengthening of the Markovian coupling notion that ensures that the sticking operation will work.

\end{document}